\documentclass[a4paper,UKenglish,cleveref,autoref,thm-restate]{lipics-v2021}

\hideLIPIcs

\usepackage{amsmath}
\usepackage{amsfonts}
\usepackage{amssymb}
\usepackage{graphicx}
\usepackage{doi}
\usepackage{amsthm}
\usepackage{xcolor}
\usepackage{todonotes}
\usepackage{hyperref}
\usepackage{xfrac}
\usepackage{mathtools}
\usepackage{algorithm}
\usepackage{algpseudocodex}

\newtheorem{question}[theorem]{Question}

\title{A Near-Linear-Time Algorithm for Finding a Well-Spread Perfect Matching in Bridgeless Cubic Graphs}

\titlerunning{A Near-Linear-Time Algorithm for Finding a WSPM in Bridgeless Cubic Graphs}

\author{Babak Ghanbari}
{Computer Science Institute of Charles University, Prague, Czech Republic}
{babak@iuuk.mff.cuni.cz}
{}{%
Supported by grant 25-16627S of the Czech Science Foundation.
}

\author{Robert \v{S}\'{a}mal}
{Computer Science Institute of Charles University, Prague, Czech Republic}
{samal@iuuk.mff.cuni.cz}
{}{%
Supported by grant 25-16627S of the Czech Science Foundation.
}

\authorrunning{B. Ghanbari and R. \v{S}\'{a}mal}

\Copyright{Babak Ghanbari and Robert \v{S}\'{a}mal}

\keywords{bridgeless cubic graphs, well-spread perfect matching, edge cuts, cactus representation}

\category{}
\relatedversion{}

\acknowledgements{We thank Matt DeVos for helpful discussion.}

\nolinenumbers

\makeatletter
\def\subjclassHeading{}
\gdef\@ccsdescString{\relax}
\makeatother

\begin{document}

\maketitle

\begin{abstract}
  We present a near-linear-time algorithm that, given a bridgeless cubic graph, finds a perfect matching intersecting every 3-edge-cut in exactly one
  edge. This improves over a cubic algorithm of Boyd et al. for the same problem, and over our previous algorithm, which worked only for 3-edge-connected
  graphs. The main ingredient is a cactus representation of the 2-edge-cuts, together with an efficient update procedure under 2-cut reductions.
\end{abstract}

\section{Introduction}
It is well known \cite{JuliusPetersen} that every bridgeless cubic (i.e., 3-regular) graph has a perfect matching. 
Edmonds \cite{Edmonds_1965} was the first to give a polynomial-time algorithm for finding a perfect matching (in any graph admitting one). 
The best algorithm for finding a perfect matching in bridgeless cubic graphs is currently due to 
Diks and Sta\'nczyk \cite{Diks} with time complexity of $O(n \log^2 n)$. 

Less well known is the fact that a bridgeless cubic graph has a perfect matching that does not contain any 3-edge-cut entirely: 
If $C$ is a set of three edges forming a minimal edge-cut, then every perfect matching $M$ satisfies $|M\cap C| = 1$ or $3$. 
There exists a perfect matching $M$ for which this intersection has size~$1$ for every 3-edge-cut $C$; we call such a matching a 
\emph{well-spread perfect matching}, or \emph{WSPM} for short. The existence of a WSPM was proved 
by Kaiser, Kr{\'a}l’ and Norine~\cite{Kaiser} (and used to study how large fraction of the edge set can be covered by two, three, etc. 
perfect matchings). 

A completely different motivation for the same question was given by Boyd, Iwata and Takazawa~\cite{boyd}. Their focus is on the 
complement of the perfect matching---a 2-factor in the given graph that intersects every 3-edge-cut. They provide a polynomial-time
algorithm to find such a 2-factor and point out that a Hamiltonian cycle is also such a 2-factor. Thus, their result can be understood 
as an approximation approach to find Hamiltonian cycle (TSP tours) in cubic graphs---an NP-complete problem and, indeed, 
such a cycle need not exist in a cubic bridgeless graph. They also showed that their algorithm helps to achieve a better approximation 
algorithm for a smallest 2-edge-connected subgraph. 

In our previous paper \cite{GS25} we improved the $O(n^3)$ algorithm of Boyd et al.\ to time $O(n \log^4 n)$, but with a significant 
limitation: our algorithm only worked on 3-edge-connected graphs. Indeed, the main idea in that paper was to use a tree representation 
of all 3-edge-cuts and efficiently update it. In this paper we remove this limitation and present an algorithm to find a WSPM in any 
bridgeless cubic graph. (For graphs with bridges, no perfect matching may exist at all, let alone well-spread one.)
The challenge is that in the presence of 2-edge-cuts, the 3-edge-cuts no longer have a tree-like structure and even their number can be 
exponential. Also, the structure of 2-edge-cuts is more complicated than the structure of 3-edge-cuts in a 3-edge-connected graph. 

The algorithm uses a prescribed-edge version of the 3-edge-connected algorithm: in a 3-edge-connected cubic graph, we can find a WSPM either containing or avoiding a
specified edge. This allows the matchings of the reduced pieces to be glued consistently.
In addition, we use the cactus representation of 2-edge-cuts (and its efficient update) to deal with the general case. 
The main technical contribution of the present paper is to show that
the cactus representation can be updated throughout the reduction process while preserving enough information to glue prescribed well-spread perfect
matchings in the pieces.
The main result we obtain in this paper is the following: 

\begin{restatable}{theorem}{mainthm}\label{thm:main}
Algorithm~\ref{alg:wspm} returns a well-spread perfect matching in a bridgeless cubic graph on $n$ vertices in time $O(n\log^4 n)$.
\end{restatable}

This also strengthens our result from \cite{Ghanbari}: the famous Cycle Double Cover conjecture \cite{szekeres_1973, seymour1979sums} says (in one equivalent formulation) that 
every bridgeless cubic graph has a surface embedding with no singular edges (edges with the same face on both sides). We can find an approximation of
this: an embedding with at most $n/10$ singular edges. We can do this in near-linear time, now for any bridgeless cubic graph. 

\section{2-edge-cuts}

Graphs may have parallel edges but no loops; a graph in which every vertex has degree $3$ is called \emph{cubic}.  
A \emph{cycle} is a connected $2$-regular graph.  
A \emph{bridge} in a graph $G$ is an edge whose removal increases the number of connected components of $G$; equivalently, it is an edge that is not contained in any cycle of $G$.  
A graph is \emph{bridgeless} if it contains no bridges.  
A connected graph is \emph{$k$-edge-connected} if it remains connected after the removal of any set of fewer than $k$ edges.

Let $G$ be a bridgeless graph. Two edges $e_1$ and $e_2$ are called \emph{equivalent}, 
written $e_1 \sim e_2$, if $e_1=e_2$ or $\{e_1,e_2\}$ is a $2$-edge-cut of $G$.  
If $e_1 \sim e_2$, then for every cycle $C$ of $G$, we have $e_1 \in E(C)$ if and only if $e_2 \in E(C)$.  
In particular, $\sim$ is an equivalence relation on $E(G)$.

The following lemma is straightforward and well-known.

\begin{lemma}\label{lem:two_edge_cut_matching}
    Let $G$ be a bridgeless cubic graph. If $e_1 \sim e_2$, then for every perfect matching $M$ of 
    $G$, either both $e_1$ and $e_2$ belong to $M$, or neither of them does.
\end{lemma}

A \emph{well-spread} perfect matching in a bridgeless cubic graph $G$ is a perfect matching $M$ that contains exactly
one edge from each $3$-edge-cut of $G$. 
A \emph{$k$-edge-cut} is an inclusion-minimal edge cut of size $k$, that is 
a set of $k$ edges whose removal increases the number of connected components of~$G$
and each of them is needed.


Let $\{e_1, e_2\}$ be a $2$-edge-cut where $e_1 = u_1v_1$, $e_2=u_2v_2$. We define a \emph{$2$-cut reduction} (more specifically, $\{e_1, e_2\}$-reduction) of a graph $G$ as follows. Construct a bridgeless cubic graph $H$ from $G$ by removing $e_1$ and $e_2$ and adding two new edges
$e'_1 = u_1u_2$ and $e'_2=v_1v_2$ to the edges of $G$. Then, $V(H) = V(G)$ 
and $E(H) = (E(G)\backslash\{e_1 = u_1v_1, e_2=u_2v_2\}) \cup \{e'_1 = u_1u_2, e'_2=v_1v_2\}$. Therefore, if $H_1$ and $H_2$ are the components of
$G-\{e_1,e_2\}$ with $u_1,u_2\in V(H_1)$ and $v_1,v_2\in V(H_2)$, then $H = (H_1 + e'_1) \cup (H_2 + e'_2)$ is a $2$-cut reduction of $G$. See the
following figure. 

\begin{center}
    \includegraphics[scale=0.18]{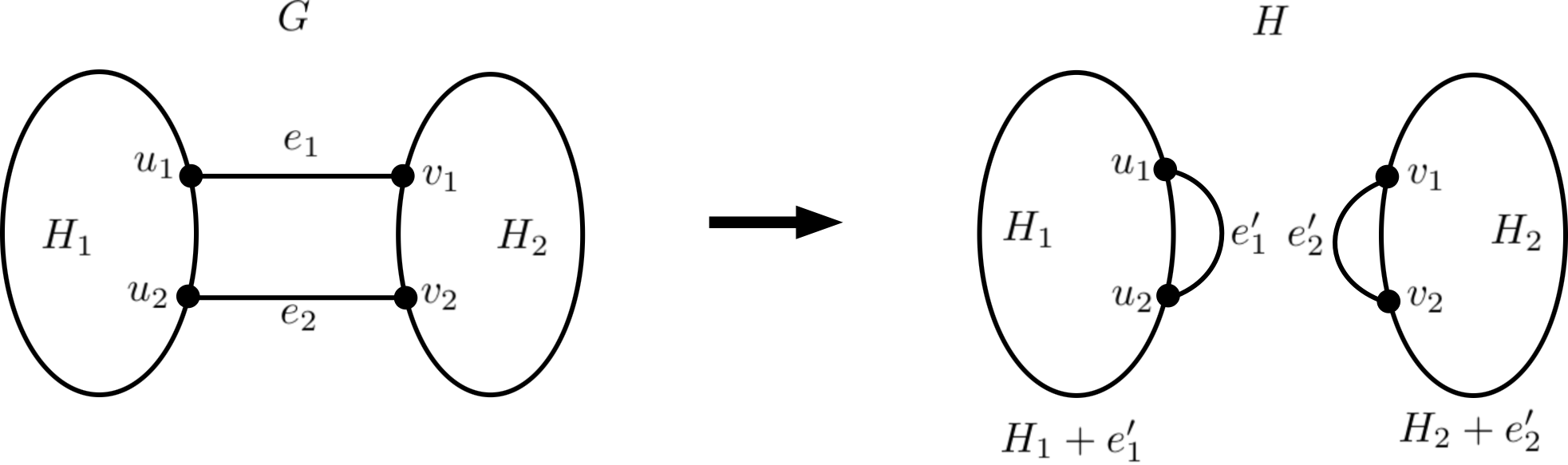}
\end{center}

\begin{lemma}\label{lem:2cut-reduction-bridgeless}
A $2$-cut reduction of a bridgeless cubic graph is again a bridgeless cubic graph.
\end{lemma}

\begin{proof}
Let $G$ be a bridgeless cubic graph and let $\{e_1,e_2\}$ be a $2$-edge-cut,
where
$e_1=u_1v_1$, $e_2=u_2v_2.$
Let $G-\{e_1,e_2\}$ have components $H_1$ and $H_2$ such that
$u_1,u_2\in V(H_1)$ and $v_1,v_2\in V(H_2)$.
The $2$-cut reduction replaces $e_1$ and $e_2$ with
$e_1' = u_1u_2$ and
$e_2' = v_1v_2$, and
the resulting graph is $H=(H_1+e_1')\cup(H_2+e_2')$.
Since $G$ is cubic, every vertex has degree $3$.
Removing $e_1$ and $e_2$ decreases the degree of each of
$u_1,u_2,v_1,v_2$ by $1$.
Adding $e_1'$ increases the degree of $u_1$ and $u_2$ by $1$,
and adding $e_2'$ increases the degree of $v_1$ and $v_2$ by $1$.
All other vertices are unchanged.
Hence every vertex of $H$ has degree $3$, so $H$ is cubic.

Now we show that no edge of $H$ is a bridge.
First consider any original edge $f$ of $H_1$ (the same argument applies
to $H_2$).
If $f$ were a bridge of $H$, then it would also be a bridge of
$H_1+e_1'$.
We will show that it would also be a bridge of~$G$. 
Any path in $G-f$ connecting the two sides of $f$ either stays in $H_1$, or enters $H_2$ through one of $e_1,e_2$ and returns through the other.
The former is impossible as we are assuming $f$ is a bridge in~$H$. 
In the latter case, the subpath through $H_2$ can be replaced by the new edge $e'_1$. Hence such a path would also exist in $H_1+e'_1-f$,
contradicting that $f$ is a bridge there.

It remains to show that the new edges $e_1'$ and $e_2'$ are not bridges.
Since $G$ is bridgeless, the edge $e_1$ lies on a cycle $C$ in $G$.
Because $\{e_1,e_2\}$ is a $2$-edge-cut, any cycle containing $e_1$
must also contain $e_2$.
Removing $e_1$ and $e_2$ from $C$ leaves a $u_1$--$u_2$ path entirely
in $H_1$ and a $v_1$--$v_2$ path entirely in $H_2$.
Adding $e_1'$ closes the path in $H_1$ to form a cycle containing $e_1'$,
and adding $e_2'$ closes the path in $H_2$ to form a cycle containing
$e_2'$.
Therefore, both $e_1'$ and $e_2'$ lie on cycles in $H$ and are not bridges.
\end{proof}

\begin{lemma}\label{lem:2-cut-reduction}
Let $G$ be a bridgeless graph and let $\{e_1,e_2\}$ and $\{e_2,e_3\}$, and $\{e_1, e_3\}$ be
$2$-edge-cuts in $G$. Perform the $\{e_2,e_3\}$-reduction of $G$, and let
$G_1$ be the component of the resulting graph that contains $e_1$. Let
$e'_2$ be the new edge added to $G_1$ by the reduction.
Then $\{e_1,e'_2\}$ is a $2$-edge-cut in $G_1$.
\end{lemma}

\begin{proof}
Let $e_1=u_1v_1$, $e_2=u_2v_2$, and $e_3=u_3v_3$.
Let $C_1$ be the connected component of $G-\{e_1,e_3\}$ containing $u_1$ and $u_3$,
let $C_2$ be the connected component of $G-\{e_1,e_2\}$ containing $v_1$ and $u_2$,
and let $C_3$ be the connected component of $G-\{e_2,e_3\}$ containing $v_2$ and $v_3$
(see Figure~\ref{fig:2-cut-preservation}).

Since $\{e_2,e_3\}$ is a $2$-edge-cut, the graph $G-\{e_2,e_3\}$ has two
components. In the $\{e_2,e_3\}$-reduction, the component containing $e_1$
receives a new edge $e'_2=u_2u_3$, while the other component receives a new
edge $e'_3=v_2v_3$. Let $G_1$ be the component containing $e_1$.

Suppose that $\{e_1,e'_2\}$ is not a $2$-edge-cut in $G_1$.
Then there exists a $u_1$--$v_1$ path $P_1$ in $G_1-\{e_1,e'_2\}$.
The path $P_1$ avoids the new edge $e'_2$ and also avoids $e_1$. Therefore, $P_1$ is a $u_1$--$v_1$ path in $G-\{e_1\}$.

By the definition of $C_1,C_2,C_3$, there exists a $v_1$--$u_2$ path $P_2$ in $C_2$,
a $v_2$--$v_3$ path $P_3$ in $C_3$, and a $u_3$--$u_1$ path $P_4$ in $C_1$. Note that a $u_3$--$u_1$ path exists otherwise $e_1$ is a bridge.
Each of the paths $P_2,P_3,P_4$ avoids the edges $e_1$ and $e_3$ by construction,
and $P_1$ avoids $e_1$ by assumption.
Now the concatenation
$P_4 \;-\; P_1 \;-\; P_2 \;-\; e_2 \;-\; P_3$
is a $u_3$--$v_3$ path in $G-\{e_1,e_3\}$.
This contradicts the fact that $\{e_1,e_3\}$ is a $2$-edge-cut of $G$.
Therefore $\{e_1,e'_2\}$ is a $2$-edge-cut in $G_1$.
\end{proof}

\begin{figure}
    \centering
     \includegraphics[scale=0.14]{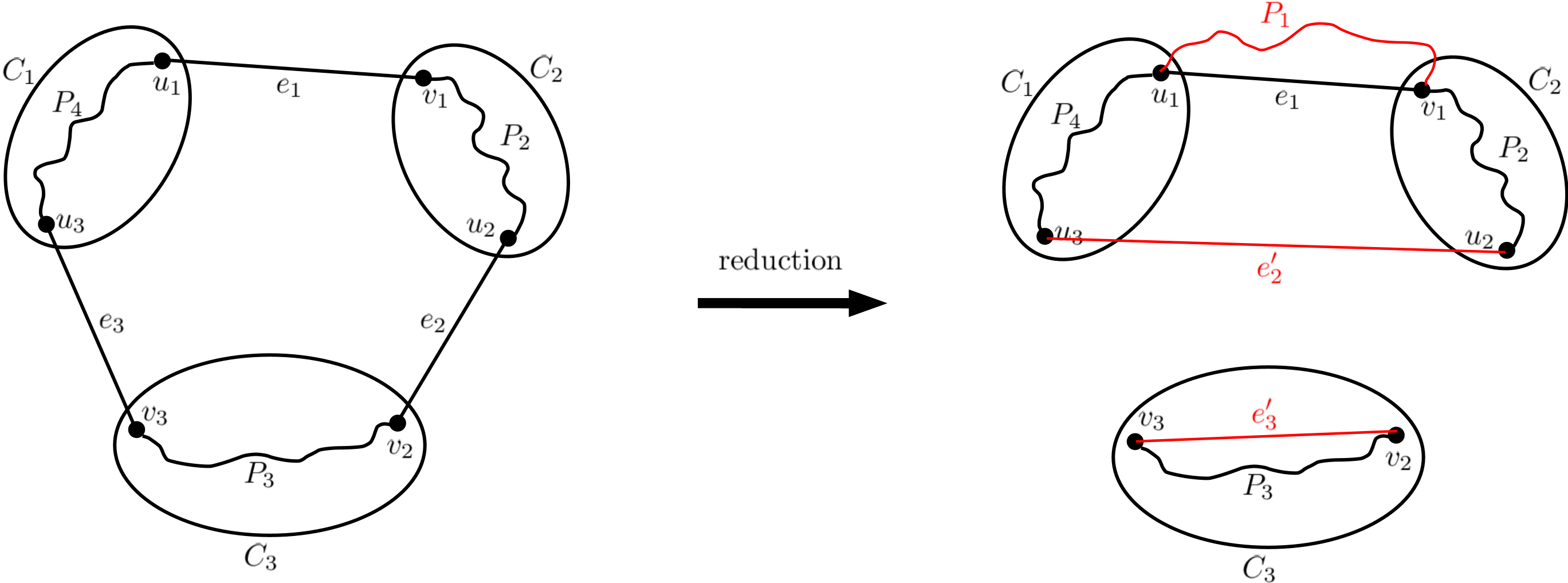}
    \caption{Demonstration of the paths $P_1, P_2, P_3$ and $P_4$ in the proof of Lemma~\ref{lem:2-cut-reduction}.}
    \label{fig:2-cut-preservation}
\end{figure}

\begin{lemma}
    Let $G$ be a graph, and $\{e_1, e_2\}$ be a 2-edge-cut in $G$. Suppose that $e_1$, $f_1$, and $f_2$ form a 3-edge-cut in $G$. Then $e_2$, $f_1$, and $f_2$ form a 3-edge-cut in $G$ as well.
\end{lemma}
\begin{proof}
    Let $f_1 = u_1v_1$, $f_2 = u_2v_2$, $e_1 = x_1y_1$ and $e_2 = x_2y_2$. Let $A$ and $B$ be the components of $G - \{e_1, f_1, f_2\}$ where $e_2 \in B$. Suppose that 
    $\{e_2, f_1, f_2\}$ is not a 3-edge-cut. Then, there exists a $x_2$--$y_2$ path $P$ in 
    $G - \{e_2, f_1, f_2\}$. Since only $e_1, f_1, f_2$ are the edges between $A$ and $B$ and both ends of $e_2$ are in $B$,
    the path $P$ cannot use $e_1$. Therefore, $P$ is also a $x_2$--$y_2$ path in 
    $G - \{e_1, e_2\}$. This is a contradiction to $\{e_1, e_2\}$ being a 2-edge-cut.
\end{proof}

The following theorem is the main tool to reduce the problem to that for 3-edge-connected graphs.
The proof is somewhat tedious, so we provide it in the appendix. 

\begin{restatable}{theorem}{thmgluing}\label{thm:2cut-gluing}
Let $G$ be a bridgeless cubic graph and let $\{e_1,e_2\}$ be a $2$-edge-cut of $G$.
Let $G_1$ and $G_2$ be the graphs obtained from the corresponding $2$-cut reduction,
and let $e'_1\in E(G_1)$ and $e'_2\in E(G_2)$ be the new edges created by the reduction.
Let $M_1$ and $M_2$ be well-spread perfect matchings of $G_1$ and $G_2$ that agree on
$e'_1$ and $e'_2$. Then $G$ has a well-spread perfect matching~$M$, that we call
a \emph{gluing} of $M_1$ and $M_2$. 
Moreover, $M$ can be obtained efficiently: 
it equals $M_1 \cup M_2$ or $(M_1\setminus\{e'_1\}) \cup (M_2\setminus\{e'_2\}) \cup \{e_1,e_2\}$
depending on whether $e'_1 \in M_1$. 
\end{restatable}

\section{Cactus representation of 2-cuts}

Recall that for a connected graph $G$ and a nonempty proper subset $X \subset V(G)$, the set
$
\delta_G(X)=\{uv\in E(G)\mid u\in X,\ v\in V(G)\setminus X\}
$
is called an \emph{edge cut} of $G$.
A \emph{minimum cut} of $G$ is an edge cut of minimum size.
The \emph{edge-connectivity} $\lambda(G)$ is the size of a minimum cut, and the \emph{local edge-connectivity} $\lambda(u, v)$ of two vertices $u$, $v$ is the size of a smallest edge cut disconnecting $u$ from $v$. An inclusion-wise maximal subset $X\subset V(G)$ such that $\lambda(u, v)\geq k$ for all $u, v \in X$ is called a \emph{$k$-edge-connected component} of $G$. 

Let $C(G)$ denote the set of all minimum cuts of $G$. A \emph{cactus} is a connected graph in which each edge is contained in at most one cycle. We allow 2-cycles but not loops.

For cacti in which every cycle has length at least $3$, it is known~\cite[P.160]{west2001introduction} that the number of edges is at most $\left\lfloor \frac{3(n-1)}{2} \right\rfloor$. However, this bound does not apply in our setting, since we allow $2$-cycles. Allowing $2$-cycles increases the maximum possible number of edges: the extremal example is obtained by taking $n-1$ copies of a $2$-cycle sharing a single common vertex. We therefore use the following observation.

\begin{observation}\label{CactusSize}
Let $T$ be a cactus with $n$ vertices, where $2$-cycles are allowed but loops are not. Then $|E(T)| \le 2(n-1)$.
\end{observation}

\begin{definition}
Let $G$ be a graph and $T$ be a cactus. A pair $(T, \varphi)$ with $\varphi: V(G) \to V(T)$ is called a cactus representation (or cactus model) for the graph $G$ if it meets the following criteria:
    \begin{enumerate}
        \item For every minimum cut $\{S, V(T) \backslash S\} \in C(T)$, the cut $\{X, \bar{X}\}$ in $G$ defined by
        $X = \{ u \in V(G) \mid \varphi(u) \in S \}$, and
        $\bar{X} = V(G) \backslash X$
        is a minimum cut in $G$.
        \item Conversely, for every minimum cut $\{X, \bar{X}\}$ in $G$, there exists a minimum cut\\
        $\{S, V(T) \backslash S\} \in C(T)$ such that
        $X = \{ u \in V(G) \mid \varphi(u) \in S \}$,
        and $\bar{X} = V(G) \backslash X$.
    \end{enumerate}
\end{definition}

Let $G$ be a bridgeless graph with at least one $2$-edge-cut.
After contracting each $3$-edge-connected component of $G$ to a single node, one obtains a cactus $T$
representing all 2-edge-cuts of $G$~\cite[Section 2.3]{Nagamochi}.
In this case, the nodes of $T$ correspond to the $3$-edge-connected components of $G$,
and no empty nodes are needed~\cite[Section 2.3.5]{Nagamochi}.
Note that the sets $\{\varphi^{-1}(x):x\in V(T)\}$ form a partition of $V(G)$,
but a node $x\in V(T)$ need not correspond to an induced subgraph
$G[\varphi^{-1}(x)]$ that is itself $3$-edge-connected.
Moreover, for two edges joining distinct contracted components, the pair forms a
$2$-edge-cut in $G$ if and only if the corresponding edges of $T$ lie on a common cycle~\cite{Nagamochi} (see Figure~\ref{fig:CactusExample}).

\begin{figure}
    \begin{center}
        \includegraphics[scale=0.15]{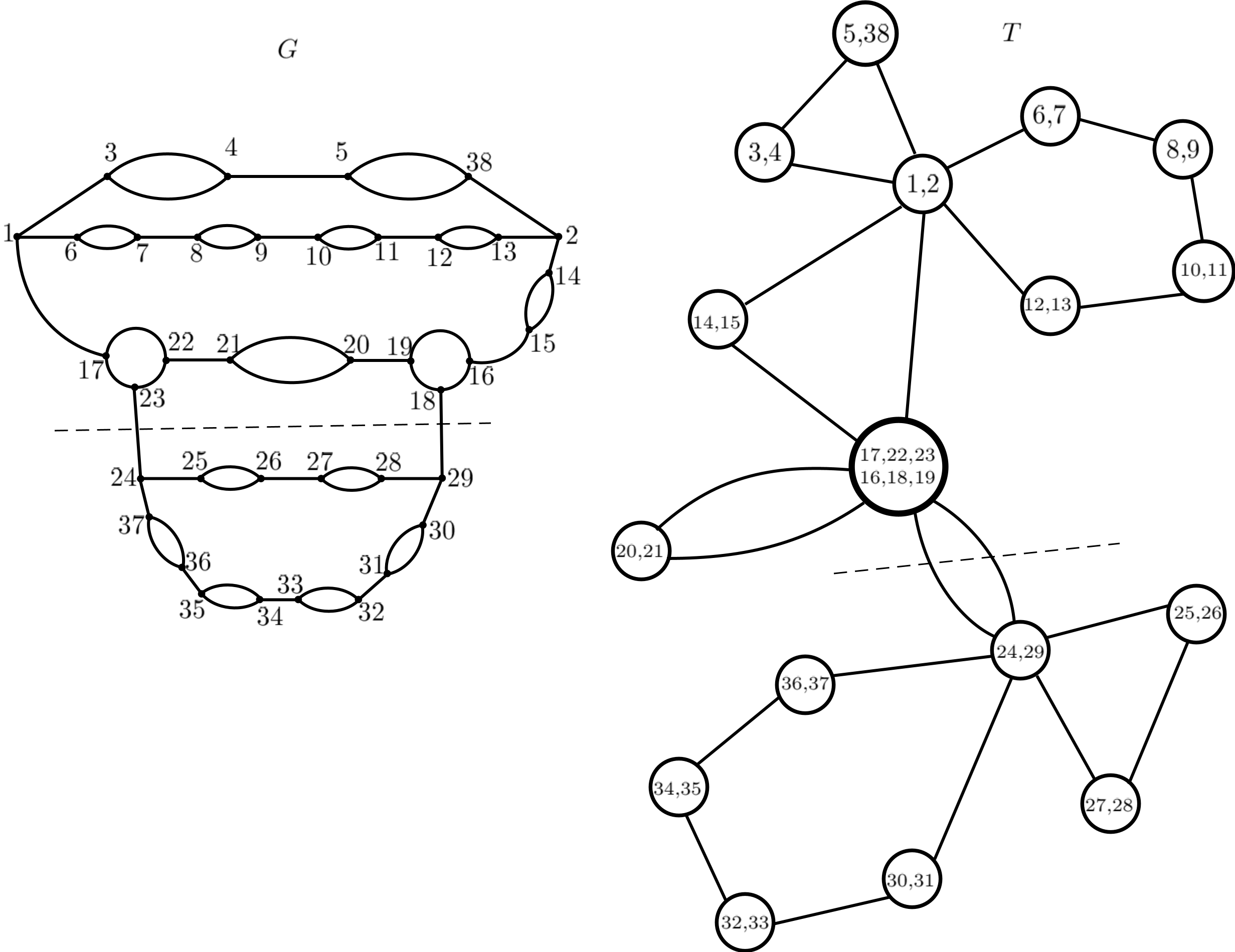}
    \end{center}
    \caption{A cubic graph $G$ and its cactus representation of the 2-edge-cuts.
    Dashed lines represent a 2-edge-cut in $G$ and its representation in $T$.
    }
    \label{fig:CactusExample}
\end{figure}

For a mapping $\varphi:V(G)\to V(T)$, an edge $e=uv\in E(G)$ is called
\emph{external} if $\varphi(u)\neq \varphi(v)$.

\if
For every external edge $e=uv$, we denote by $\psi(e)$ the edge of $T$ joining
$\varphi(u)$ and $\varphi(v)$.
Thus the map $\psi$ is only auxiliary notation that records which cactus edge
corresponds to a given external edge of $G$.
In particular, for any two distinct external edges $e,f\in E(G)$,
$\{e,f\}$ is a $2$-edge-cut of $G$ if and only if
$\psi(e)$ and $\psi(f)$ lie on a common cycle of $T$~\cite{Nagamochi}.
Figure~\ref{fig:CactusExample} illustrates this construction.
\fi

\begin{lemma}\label{lem:cactus-reduction}
Let $G$ be a bridgeless cubic graph with at least one $2$-edge-cut and let $(T,\varphi)$ be a cactus representation of all $2$-edge-cuts of $G$. Let $\{e_1,e_2\}$ be a $2$-edge-cut of $G$, and let $\{a,b\}$ be the corresponding $2$-edge-cut in $T$. Perform the $\{e_1,e_2\}$-reduction in $G$, obtaining graphs $G_1$ and $G_2$ with new edges $e'_1$ and $e'_2$. Perform the $\{a,b\}$-reduction in $T$, obtaining graphs $T_1$ and $T_2$. If a loop is created in $T_1$ or $T_2$, delete it. Then 
$(T_1,\varphi_1 = \varphi|_{V(G_1)})$ and $(T_2,\varphi_2 = \varphi|_{V(G_2)})$
are cactus representations of all $2$-edge-cuts of $G_1$ and $G_2$, respectively.
\end{lemma}
\begin{proof}
We prove the statement for $G_1$; the proof for $G_2$ is symmetric.

Let $H_1,H_2$ be the components of $G-\{e_1,e_2\}$, so that
$G_1=H_1+e'_1$ and $G_2=H_2+e'_2$, where $e'_1$ and $e'_2$ are the new edges.
Let $A_1,A_2$ be the components of $T-\{a,b\}$, and let $a'$ and $b'$ be the new edges added to $A_1$ and $A_2$, respectively, by the $\{a,b\}$-reduction. Delete $a'$ or $b'$ if it is a loop.

By Lemma~\ref{lem:2cut-reduction-bridgeless}, $G_1$ is bridgeless cubic, so its minimum cuts are
exactly its $2$-edge-cuts. Thus it suffices to show that for any external edges of $G_1$, the pair
forms a $2$-edge-cut in $G_1$ if and only if the corresponding edges of $T_1$ lie on a common cycle.

\medskip
\noindent
\textbf{Case 1:} $f,g\in E(G_1)\setminus\{e'_1\}$.
Then $f,g\subseteq E(H_1)\subseteq E(G)$. We claim that
$\{f,g\}$ is a $2$-edge-cut in $G_1$ if and only if $\{f,g\}$ is a $2$-edge-cut in $G$.
Indeed, any path in $G-\{f,g\}$ that traverses $H_2$ must enter through one of $e_1,e_2$
and leave through the other, so each maximal subpath through $H_2$ can be replaced by the
single edge $e'_1$, yielding a path in $G_1-\{f,g\}$. Conversely, any use of $e'_1$ in a path
of $G_1-\{f,g\}$ can be expanded to a path through $e_1$, then inside $H_2$, and then through
$e_2$, yielding a path in $G-\{f,g\}$. Hence the claim follows.

Since $(T,\varphi)$ is a cactus representation of all $2$-edge-cuts of $G$, the latter is equivalent
to the corresponding edges of $T$ lying on a common cycle. As neither corresponding edge is $a$ or $b$,
this is equivalent to the corresponding edges in $T_1$ lying on a common cycle. Therefore
$\{f,g\}$ is a $2$-edge-cut in $G_1$ if and only if the corresponding edges lie on a common cycle in $T_1$.

\medskip
\noindent
\textbf{Case 2:} one edge is $e'_1$, say $\{e'_1,f\}$.
Let $c$ be the edge of $T_1$ corresponding to $f$.
Assume first that $c$ and $a'$ lie on a common cycle in $T_1$. Then the edge of $T$
corresponding to $f$ lies on the cycle containing $a$ and $b$. Since $(T,\varphi)$ represents
all $2$-edge-cuts of $G$, both $\{e_1,f\}$ and $\{e_2,f\}$ are $2$-edge-cuts in $G$. By
Lemma~\ref{lem:2-cut-reduction}, after performing the $\{e_1,e_2\}$-reduction, it follows that
$\{e'_1,f\}$ is a $2$-edge-cut in $G_1$.

Conversely, assume that $\{e'_1,f\}$ is a $2$-edge-cut in $G_1$. We show that $\{e_1,f\}$ is
a $2$-edge-cut in $G$. Suppose not. Then there exists in $G-\{e_1,f\}$ a path joining the two
ends of $e_1$. Since $\{e_1,e_2\}$ is a $2$-edge-cut, every such path must use $e_2$, and its
subpath inside $H_1$ yields a path in $G_1-\{e'_1,f\}$ joining the ends of $e'_1$, a contradiction.
Hence $\{e_1,f\}$ is a $2$-edge-cut in $G$. Since $\{e_1,e_2\}$ is also a $2$-edge-cut, the
corresponding edge of $T$ lies on the same cycle as $a$ and $b$, and therefore $c$ lies on a
common cycle with $a'$ in $T_1$.

Thus also in this case $\{e'_1,f\}$ is a $2$-edge-cut in $G_1$ if and only if the corresponding
edges lie on a common cycle in $T_1$.

Finally, if $a'$ is a loop, we delete it. This does not affect the set of represented
$2$-edge-cuts, since a loop belongs to no cycle of the cactus. Therefore
$(T_1,\varphi_1)$ is a cactus representation of all $2$-edge-cuts of $G_1$. The proof for
$(T_2,\varphi_2)$ is identical.
\end{proof}

\begin{figure}
    \begin{center}
        \includegraphics[scale=0.14]{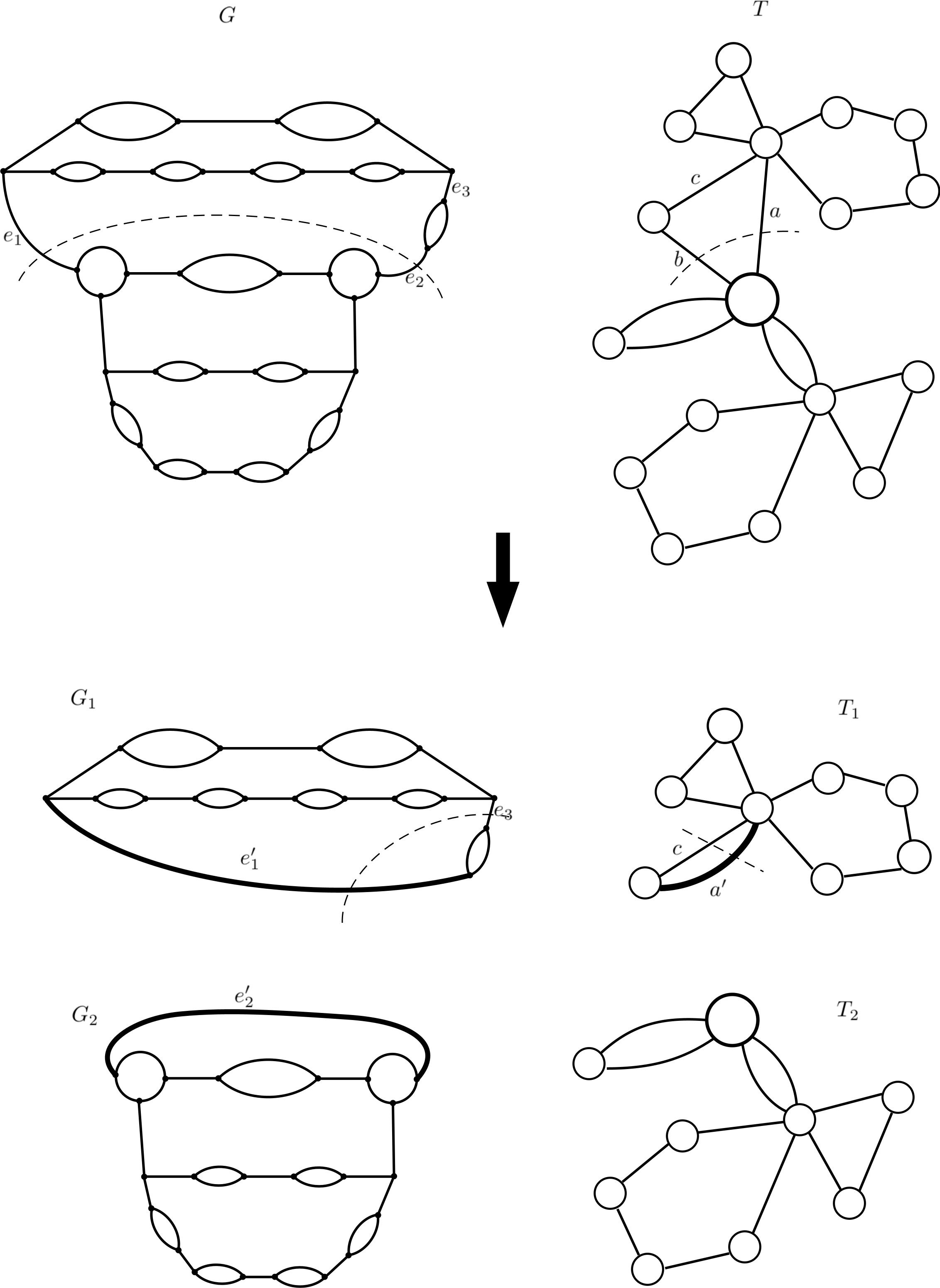}
    \end{center}
    \caption{Illustration of the $\{e_1,e_2\}$-reduction from Lemma~\ref{lem:cactus-reduction} and the corresponding update of the cactus representation. The edges $e_1$ and $e_2$ form a $2$-edge-cut in $G$, corresponding to the cactus edges $a$ and $b$ in $T$. After the reduction, the graph splits into $G_1$ and $G_2$ with new edges $e'_1$ and $e'_2$, and the cactus is updated to $T_1$ and $T_2$ with new edges $a'$ and $b'$. The edge $b'$ becomes a loop in $T_2$ and is therefore deleted. Edges such as $c$ illustrate that the remaining structure is preserved.}
    \label{fig:CactusReductionLemma}
\end{figure}

\begin{theorem}\label{thm:edge_forced_wspm}
Let $G$ be a $3$-edge-connected cubic graph and let $e\in E(G)$.
Then there exists a well-spread perfect matching $M$ such that $e\in M$.
Moreover, such a matching can be computed in time $O(n\log^4 n)$.
\end{theorem}

\begin{proof}
Let $e=uv\in E(G)$. We use Algorithm~1 from~\cite{GS25}.
Let $(T,\varphi)$ be the tree representation of the nontrivial 3-edge-cuts. 
Since $G$ is $3$-edge-connected and cubic, every vertex of $G$ is mapped to a
leaf of $T$. In particular, both $\varphi(u)$ and $\varphi(v)$ are leaves of
$T$. Let $r$ be the unique neighbour of the leaf $\varphi(u)$, and root $T$
at $r$.

Consider the decomposition phase of Algorithm~1 of \cite{GS25}.
At each step, the algorithm separates a graph corresponding to a non-root
subtree of $T$. Since $\varphi(u)$ is a leaf whose parent is the root $r$, no
separated subtree contains $\varphi(u)$ (the subtree consisting solely of $\varphi(u)$ 
is processed only at the final step). Hence no separated graph contains the
vertex $u$. Therefore, no separated graph can contain the whole edge $e=uv$, because that
would require it to contain both endpoints of $e$.
Consequently, the edge $e$ remains present in the current graph throughout the
decomposition phase. In particular, it is present in the final graph $G'$
associated with the root $r$.
At the end of the decomposition phase, the remaining graph $G'$ (the graph
associated with the root $r$) has no non-trivial $3$-edge-cuts, and hence
is internally $4$-edge-connected (see~\cite{GS25}).
In particular, $G'$ is a bridgeless cubic graph.
By Theorem~1 of~\cite{GS25} (Schönberger's theorem),
there exists a perfect matching $M'$ of $G'$ that contains the edge $e$,
and such a matching can be found in time $O(n\log^4 n)$.

We now perform the assembly phase of Algorithm~1 of~\cite{GS25}.
At each step, we combine the current matching with a perfect matching
of a previously separated graph using Theorem~2
of~\cite{GS25}. Since all matchings agree on the corresponding
$3$-edge-cuts, the resulting matching remains perfect and well-spread.
Since $e$ was present in $G'$ and was chosen to belong to $M'$,
and the assembly procedure never removes edges from the current matching,
it follows that $e\in M$, where $M$ is the final matching obtained for $G$.
The running time is the same as in Algorithm~1 of
\cite{GS25}, which is $O(n\log^4 n)$.
Thus, there exists a well-spread perfect matching of $G$ containing $e$,
and it can be found in time $O(n\log^4 n)$.
\end{proof}

\begin{corollary}\label{thm:edge_forced_wspm_not_in}
Let $G$ be a $3$-edge-connected cubic graph and let $e\in E(G)$.
Then there exists a well-spread perfect matching $M$ such that $e\notin M$.
Moreover, such a matching can be computed in time $O(n\log^4 n)$.
\end{corollary}

\begin{proof}
Let $e=uv_1$ and $e' = uv_2$ be two edges in $E(G)$ incident at $u$. By Theorem~\ref{thm:edge_forced_wspm} there exists a well-spread perfect matching $M$ such that $e'\in M$. Since $u$ is already covered by $M$, $e\notin M$. 
\end{proof}

\section{Algorithm}
At each step we choose a vertex of degree $2$ in the cactus. 
(It exists, unless there is no 2-edge-cut left, as follows easily by Observation~\ref{CactusSize}.) 
The two cactus edges incident with this vertex correspond
to a $2$-edge-cut in the current graph. Performing the corresponding $2$-cut
reduction separates a $3$-edge-connected graph from the remainder.
Repeating this procedure gives a sequence
$ H_0=G,\ H_1,\ \dots,\ H_{k-1},$
together with $3$-edge-connected graphs
$G_1,G_2,\dots,G_k,H_k,$
where for each $i\in\{1,\dots,k\}$ the graph $H_{i-1}$ is reduced along a
$2$-edge-cut into two graphs $G_i$ and $H_i$ (for $i<k$), and in the final
step the graph $H_{k-1}$ is reduced into two $3$-edge-connected graphs
$G_k$ and $H_k$.

To achieve the required running time, we store the graph and the cactus using adjacency lists and maintain the list $L$ of degree-2 nodes of the cactus. 
For each cactus node $x$ we store the list of vertices $\varphi^{-1}(x)$. 
For each cactus edge we store the corresponding external edge of the current graph. 
Each cactus update does bounded amount of work, therefore all cactus updates and reduction records can be produced in total $O(n)$ time. 

During the reduction phase we store a list $\mathcal R$ of reduction records.
Each record consists of a tuple $(i,e_1,e_2,e'_1,e'_2)$
describing a performed $2$-cut reduction where $e_1,e_2$ are the removed edges
and $e'_1,e'_2$ are the newly created edges.

A well-spread perfect matching is then constructed backwards. First we find an
arbitrary well-spread perfect matching in one of the final graphs, say $H_k$.
Then we find a well-spread perfect matching in $G_k$ that agrees with it on
the reduction edge. By Theorem~\ref{thm:2cut-gluing}, these two matchings can be glued
to a well-spread perfect matching of $H_{k-1}$. Repeating this step backwards
produces a well-spread perfect matching of the original graph $G$.
Again, each step changes at most two edges, thus we do these updates in total $O(n)$ time.

\begin{algorithm}[H]
\caption{Well-spread perfect matching in a bridgeless cubic graph}
\label{alg:wspm}
\begin{algorithmic}[1]
\Require A bridgeless cubic graph $G$
\Ensure A well-spread perfect matching $M$

\If{$G$ is 3-edge-connected}
    \State \Return the output of Algorithm~1 of~\cite{GS25}
\Else
    \State Compute a cactus representation $(T,\varphi)$ of all $2$-edge-cuts of $G$
    \Comment{\cite{Tsin2023113760} in $O(n)$ time}
    \State Make a list $L$ of degree-$2$ nodes in $T$
    \State $H_0 \gets G$
    \State $\mathcal R \gets$ empty list
    \State $i \gets 1$

    \Comment{Part 1: reduction sequence}
    \While{$E(T)\neq\emptyset$}
        \State Choose a node $x\in L$
        \State Let $a=xy$ and $b=xz$ be the two edges of $T$ incident with $x$
        \State Let $e_1$ be an edge of $H_{i-1}$ with one end in $\varphi^{-1}(x)$ and the other in $\varphi^{-1}(y)$
        \State Let $e_2$ be an edge of $H_{i-1}$ with one end in $\varphi^{-1}(x)$ and the other in $\varphi^{-1}(z)$
        \State Perform the $\{e_1,e_2\}$-reduction of $H_{i-1}$,
        obtaining graphs $G_i$ and $H_i$ and new edges
        $e'_1\in E(H_i)$ and $e'_2\in E(G_i)$
        \State Append $(i,e_1,e_2,e'_1,e'_2)$ to $\mathcal R$
        \State Perform the $\{a,b\}$-reduction of $T$, deleting the isolated node $x$ and any loop edges created
        \State Update $L$
        \State $i \gets i+1$
    \EndWhile

    \State $k \gets i-1$

    \Comment{Part 2: backward construction}
    \State Compute any WSPM $M$ of $H_k$
    \Comment{\cite{GS25} in $O(n_0\log^4 n_0)$ time}

    \For{the records $(i,e_1,e_2,e'_1,e'_2)$ of $\mathcal R$ in reverse order}
        \If{$e'_1\in M$}
            \State Compute a WSPM $M_i$ of $G_i$ containing $e'_2$
            \Comment{Theorem~\ref{thm:edge_forced_wspm}}
        \Else
            \State Compute a WSPM $M_i$ of $G_i$ avoiding $e'_2$
            \Comment{Corollary~\ref{thm:edge_forced_wspm_not_in}}
        \EndIf
        \State Replace $M$ by the matching obtained by gluing $M$ and $M_i$
    \EndFor

    \State \Return $M$
\EndIf
\end{algorithmic}
\end{algorithm}

\mainthm*


\begin{proof}
Let $n=|V(G)|$. Since $G$ is cubic, we have $|E(G)|=O(n)$.

We first prove correctness. Let $H_0=G$. During the reduction phase, the algorithm maintains a current bridgeless cubic graph $H_{i-1}$ together with a cactus representation $(T,\varphi)$ of all $2$-edge-cuts of $H_{i-1}$. In each iteration, the algorithm chooses a vertex $x$ of degree $2$ in $T$, let $a=xy$ and $b=xz$ be the two incident cactus edges, and selects the corresponding edges $e_1,e_2$ in $H_{i-1}$. These edges form a $2$-edge-cut in $H_{i-1}$. Performing the $\{e_1,e_2\}$-reduction produces two graphs $G_i$ and $H_i$, which are both bridgeless cubic by Lemma~\ref{lem:2cut-reduction-bridgeless}. By Lemma~\ref{lem:cactus-reduction}, after performing the corresponding $\{a,b\}$-reduction in $T$, the restricted mapping yields again a cactus representation of all $2$-edge-cuts of $H_i$. Hence the invariant is preserved throughout the reduction phase.
When the reduction phase terminates, the cactus has no edges. Since it represents all $2$-edge-cuts of $H_k$, the graph $H_k$ has no $2$-edge-cuts. As $H_k$ is bridgeless, it follows that $H_k$ is $3$-edge-connected.
Now consider a graph $G_i$ separated during some reduction step. By Lemma~\ref{lem:cactus-reduction}, the corresponding cactus representation of $G_i$ has no edges, since $G_i$ is obtained from a single component created at a degree-$2$ cactus vertex. Therefore $G_i$ has no $2$-edge-cuts. As $G_i$ is bridgeless cubic, it follows that $G_i$ is $3$-edge-connected. Thus all graphs
$G_1,G_2,\dots,G_k,H_k$
are $3$-edge-connected cubic graphs.

We now prove correctness of the backward phase. First, the algorithm computes a well-spread perfect matching $M$ of $H_k$. For reduction records in reverse order, at each step we glue a well-spread perfect matching of $H_i$ with a well-spread perfect matching of $G_i$ that agrees on the reduction edges. By induction, this produces a well-spread perfect matching of $H_{i-1}$. At the end, we obtain a well-spread perfect matching of $H_0=G$.

It remains to analyze the running time. The cactus representation of all
$2$-edge-cuts of $G$ can be computed in time $O(n)$ and has $O(n)$ size. 
Each iteration of the reduction phase removes at least one edge from the cactus, and by
Observation~\ref{CactusSize} the cactus has $O(n)$ edges. Hence the reduction
phase takes $O(n)$ time.
Let $n_i=|V(G_i)|$ for $i=1,\dots,k$ and let $n_0=|V(H_k)|$. 
Since the reductions partition the original vertex set among the final pieces, we have
$$
  \sum_{i=0}^{k} n_i = n.
$$
For each of these graphs, the algorithm computes one well-spread perfect
matching with a prescribed edge. By Theorem~\ref{thm:edge_forced_wspm} and
Corollary~\ref{thm:edge_forced_wspm_not_in}, this takes time
$O(n_i\log^4 n_i)$ on a graph with $n_i$ vertices.
Since $n_i\le n$ for every $i$, we have $\log n_i\le \log n$, and therefore
$
    n_i\log^4 n_i \le n_i\log^4 n.
$
Summing over all pieces gives
$$
    \sum_{i=0}^k n_i\log^4 n_i
    \le
    \sum_{i=0}^k n_i\log^4 n
    =
    \log^4 n\sum_{i=0}^k n_i
    =
    n\log^4 n.
$$
Thus the total time spent computing matchings in the backward phase is
$O(n\log^4 n)$.

Finally, each gluing step takes constant time and there are $O(n)$ such
steps. Therefore the total running time of Algorithm~\ref{alg:wspm} is
$O(n\log^4 n)$. This completes the proof.
\end{proof}

\section{Final remarks}

We described an almost linear algorithm that finds a WSPM in a cubic bridgeless graph. 
Boyd et al. also gave a polynomial-time algorithm that in a cubic bridgeless graph 
finds a 2-factor intersecting every 3-edge-cut and every 4-edge-cut. 

\begin{question}
  Can we do that in near-linear time?
\end{question}

\bibliography{RN.bib}

\appendix
\section{Appendix}


\thmgluing*

\begin{proof}
Let $e_1=u_1v_1$ and $e_2=u_2v_2$, where
$u_1,u_2\in V(G_1)$ and $v_1,v_2\in V(G_2)$.
In the $2$-cut reduction, the edges $e_1$ and $e_2$ are replaced by
$e'_1=u_1u_2 \in E(G_1)$, and $e'_2=v_1v_2 \in E(G_2)$.

We construct a perfect matching $M$ of $G$ as follows.
If neither $e'_1\in M_1$ nor $e'_2\in M_2$, set
$M = M_1 \cup M_2$.
Otherwise both $e'_1\in M_1$ and $e'_2\in M_2$, and we set
$M = (M_1\setminus\{e'_1\}) \cup (M_2\setminus\{e'_2\}) \cup \{e_1,e_2\}$.
We call $M$ the \emph{gluing} of $M_1$ and $M_2$. In both cases every vertex of $G$ is incident with exactly one edge of $M$:
$M_1$ and $M_2$ are perfect matchings in disjoint parts of the graph,
and in the second case the removal of $e'_1, e'_2$ and insertion of $e_1, e_2$
preserves the matching at their endpoints. Therefore, $M$ is a perfect matching of $G$.

It remains to prove that $M$ is well-spread.
Let $C = \{f_1, f_2, f_3\}$ be any $3$-edge-cut of $G$. We show that $M$ contains exactly one of $f_1$, $f_2$, or $f_3$. Since $\{e_1,e_2\}$ is a $2$-edge-cut,
no $3$-edge-cut can contain both $e_1$ and $e_2$. 
Therefore, either $C$ contains one of $e_1$ and $e_2$ or it contains neither of them. By symmetry between $G_1$ and $G_2$, it suffices to consider only those
3-edge-cuts that contain at least one edge of $G_1$. The remaining
possibilities where the cut lies entirely in $G_2$ or contains $e_2$
are handled analogously. We have the following cases:
\begin{enumerate}
    \item $f_1, f_2, f_3 \in  E(G_1\backslash e'_1)$. Since $M_1$ is well-spread in $G_1$, exactly one edge of 
        $f_1$, $f_2$, or $f_3$ belongs to $M_1$, say $f_1$. By construction of $M$, the intersection of $M$ with $C$
        coincides with the intersection of $M_1$ with $C$. Therefore, $M$ only contains $f_1$.

    \item $f_2,f_3\in E(G_1\setminus e'_1)$ and $f_1=e_1$. Let $H_1$ and $H_2$ be the components of $G-\{e_1,e_2\}$.
        Let $S_1$ and $S_2$ be the components of $G-\{e_1,f_2,f_3\}$.
        Since $f_2,f_3\in E(G_1\setminus e'_1)=E(H_1)$, they are not in $H_2$.
        Because $e_2\notin\{e_1,f_2,f_3\}$, both ends of $e_2$ lie in the same component of
        $G-\{e_1,f_2,f_3\}$. Assume $u_2,v_2\in S_2$.
        Also $v_1\in S_2$, since $v_1$ is adjacent to $u_1$ only via the removed edge $e_1$ and
        $H_2$ is connected.
        We claim that in fact $V(H_2)\subseteq S_2$.
        Indeed, if some vertex of $H_2$ lay in $S_1$, then (as $H_2$ is connected) there would be a path in
        $H_2$ from that vertex to $v_1\in S_2$, which would avoid $\{e_1,f_2,f_3\}$ and hence contradict that $\{e_1,f_2,f_3\}$ is a 3-edge-cut.

        Consequently, $u_1\in S_1$ and $u_2\in S_2\cap V(G_1)$, so the added edge $e_1'=u_1u_2$ goes between
        $S_1$ and $S_2\cap V(G_1)$.
        Moreover, the only edges of $H_1$ between $S_1$ and $S_2$ are $f_2$ and $f_3$, since
        $G-\{e_1,f_2,f_3\}$ has exactly the two components $S_1,S_2$.
        Hence $\{e_1',f_2,f_3\}$ is a $3$-edge-cut of $G_1$.
        Since $M_1$ is well-spread in $G_1$, we have
        $|M_1\cap\{e_1',f_2,f_3\}|=1$.
        \smallskip
        
        Finally, we show that $M$ meets $\{e_1,f_2,f_3\}$ in exactly one edge.

        \smallskip
        \noindent\emph{If $e_1'\notin M_1$ and $e_2'\notin M_2$,} then $M=M_1\cup M_2$ by construction.
        Thus $e_1\notin M$, and since $|M_1\cap\{e_1',f_2,f_3\}|=1$ with $e_1'\notin M_1$, we get
        $|M\cap\{f_2,f_3\}|=1$. Therefore $|M\cap\{e_1,f_2,f_3\}|=1$.

        \smallskip
        \noindent\emph{If $e_1'\in M_1$ and $e_2'\in M_2$,} then the construction replaces $e_1',e_2'$ by
        $e_1,e_2$, so
        $
        M=(M_1\cup M_2)\cup\{e_1,e_2\}\setminus\{e_1',e_2'\}.
        $
        Since $e_1'\in M_1$ and $|M_1\cap\{e_1',f_2,f_3\}|=1$, we have $f_2,f_3\notin M_1$ and hence
        $f_2,f_3\notin M$. Also $e_1\in M$ by construction. Therefore again
        $|M\cap\{e_1,f_2,f_3\}|=1$.

    \item $f_1 = e_1$, $f_2\in E(G_1\setminus e'_1)$ and $f_3 \in E(G_2\setminus e'_2)$.
        We show that this case is impossible.

        Let $f_2=x_1x_2$ and $f_3=y_1y_2$.
        Since $e_2$ is not in the $3$-edge-cut $\{e_1,f_2,f_3\}$, both its ends
        must lie in the same component of $G-\{e_1,f_2,f_3\}$.
        Without loss of generality, assume that $u_2,v_2\in S_2$.
        Then
        \begin{itemize}
            \item $u_1,u_2,x_1,x_2\in H_1$,
            \item $v_1,v_2,y_1,y_2\in H_2$,
            \item $u_1,x_1,y_1\in S_1$ (otherwise, relabel the endpoints of $f_2$ and $f_3$ and exchange $G_1$ with $G_2$)
            \item $v_1,v_2,x_2,y_2,u_2\in S_2$.
        \end{itemize}

        \begin{center}
            \includegraphics[scale=0.16]{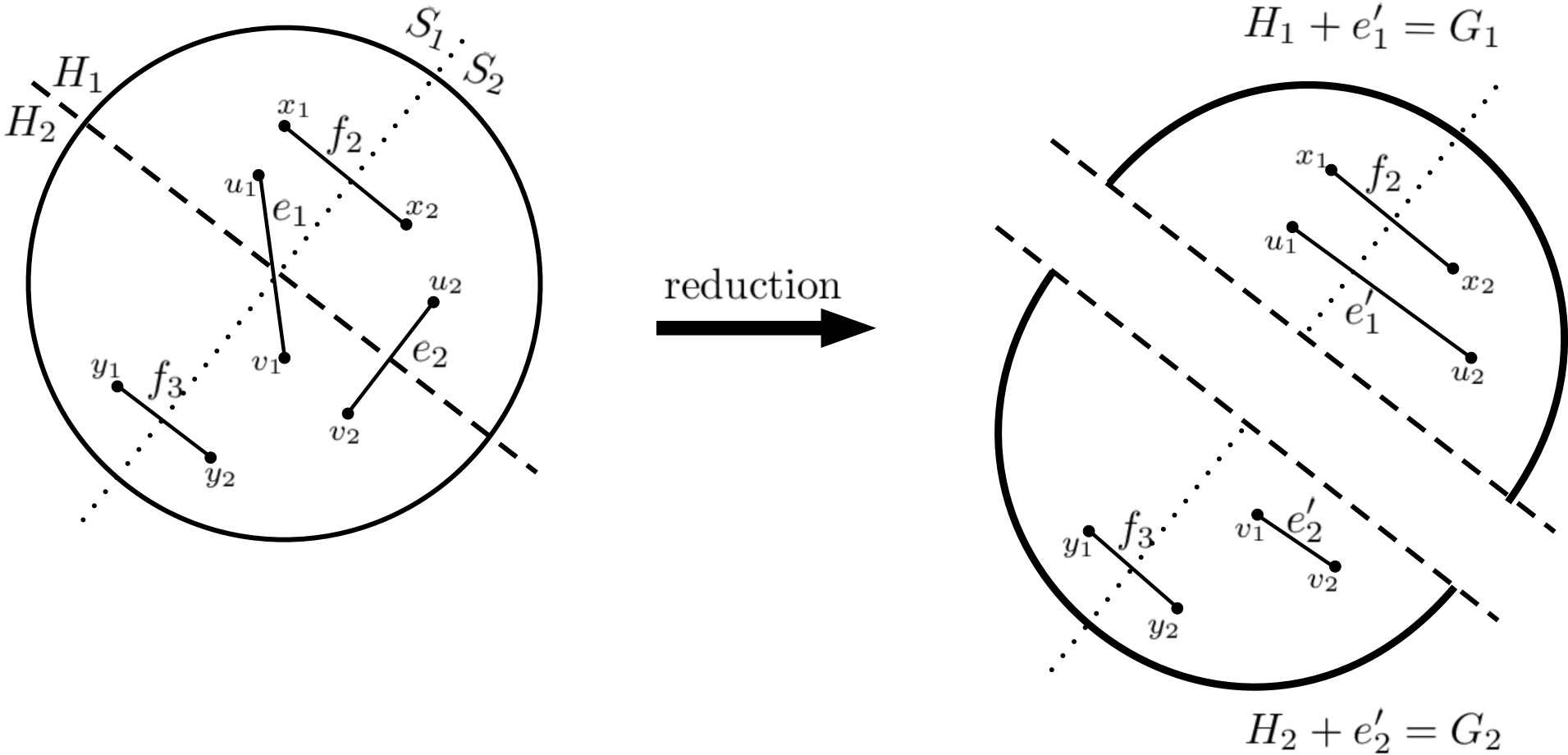}
        \end{center}

        We claim that $f_3=y_1y_2$ is a bridge in $G_2=H_2+e'_2$.
        Suppose not. Then there exists a $y_1$--$y_2$ path $P$ in
        $G_2-\{f_3\}$.

        First suppose that $e'_2\notin P$.
        Then $P$ lies entirely in $H_2-\{f_3\}$.
        Since $e_1$ and $f_2$ are not in $H_2$, the same path $P$ exists in
        $G-\{e_1,f_2,f_3\}$, which connects $y_1\in S_1$ to
        $y_2\in S_2$.
        This contradicts that $S_1$ and $S_2$ are the components of
        $G-\{e_1,f_2,f_3\}$.
        \begin{center}
            \includegraphics[scale=0.16]{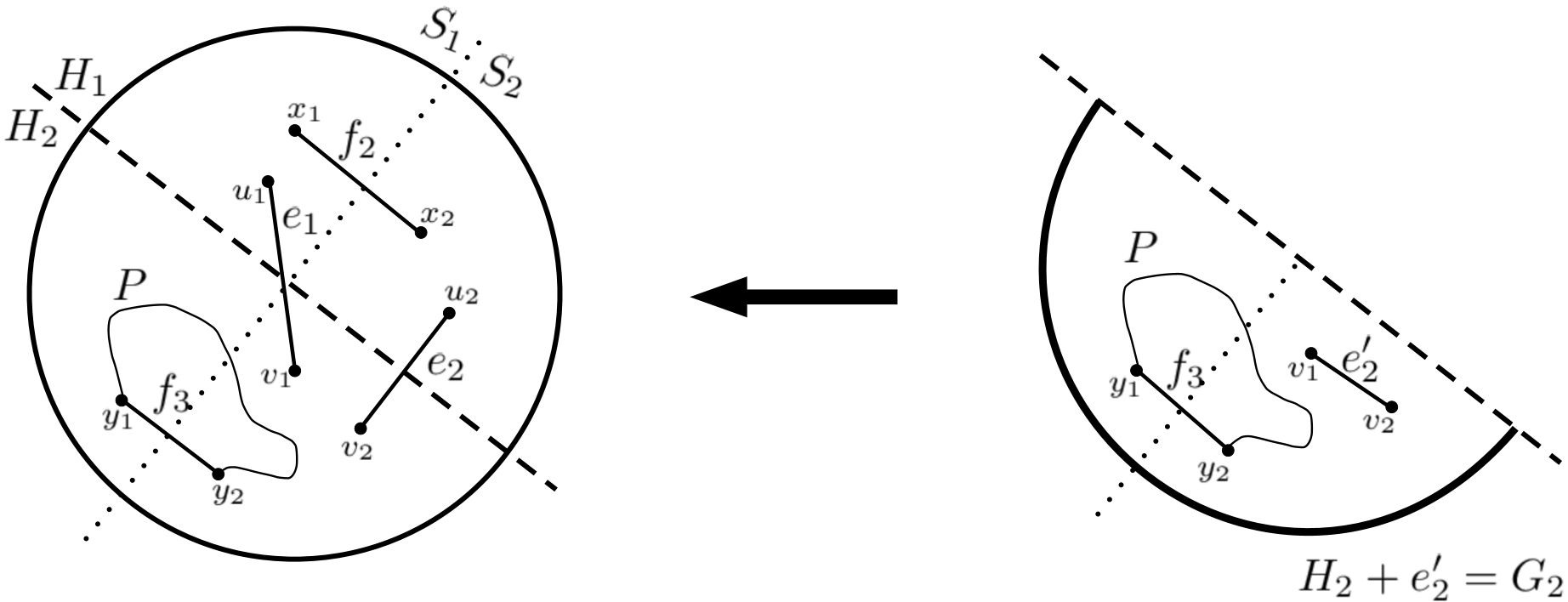}
        \end{center}
        Now suppose that $e'_2\in P$.
        Recall $e'_2=v_1v_2$ and let $Q$ be the maximal initial subpath of $P$
        from $y_1$ to the first vertex of $\{v_1,v_2\}$ encountered along $P$.
        Then $Q$ avoids $e'_2$ and hence lies entirely in $H_2-\{f_3\}$.
        Therefore $Q$ is a path in $G-\{e_1,f_2,f_3\}$.
        But $Q$ connects $y_1\in S_1$ to a vertex of $\{v_1,v_2\}\subseteq S_2$,
        again contradicting that $S_1$ and $S_2$ are the components of
        $G-\{e_1,f_2,f_3\}$.
        \begin{center}
            \includegraphics[scale=0.16]{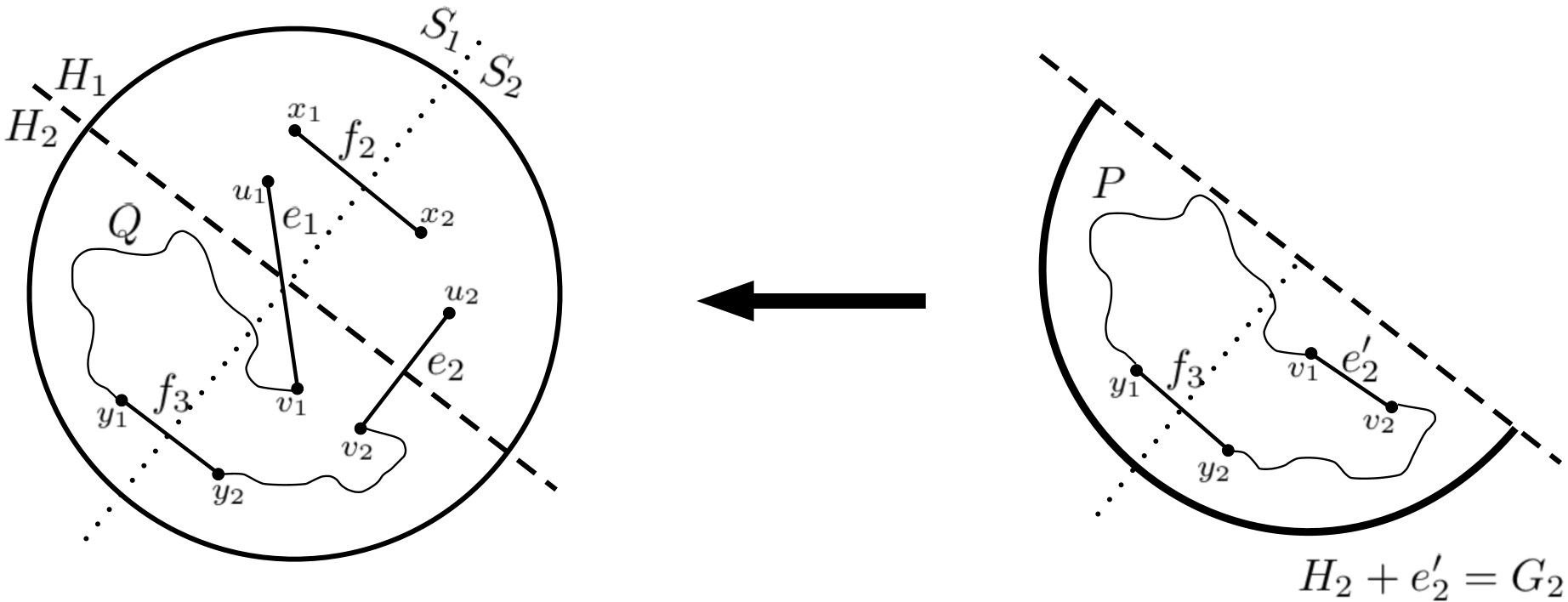}
        \end{center}

        Thus no such path $P$ exists and $f_3$ is a bridge in $G_2$,
        contradicting that $G_2$ is bridgeless.
        Therefore this case cannot occur.

    \item $f_1, f_2 \in E(H_1)$ and $f_3 \in E(H_2)$.
        Recall that $e_1=u_1v_1$ and $e_2=u_2v_2$ with $u_1,u_2\in V(H_1)$ and
        $v_1,v_2\in V(H_2)$, and that $e_1'=u_1u_2\in E(G_1)$ and
        $e_2'=v_1v_2\in E(G_2)$.

        We first claim that $v_1$ and $v_2$ lie in different components of
        $G-\{f_1,f_2,f_3\}$. Suppose for a contradiction that $v_1,v_2\in S_1$
        (the case $v_1,v_2\in S_2$ is symmetric). Since $f_3$ is the only edge of
        $\{f_1,f_2,f_3\}$ contained in $H_2$, it is the only edge of $H_2$ with one
        end in $S_1\cap V(H_2)$ and the other in $S_2\cap V(H_2)$. As $v_1,v_2\in S_1$,
        the added edge $e_2'=v_1v_2$ also has both ends in $S_1$, so it does not cross
        between $S_1\cap V(H_2)$ and $S_2\cap V(H_2)$. Therefore, $f_3$ is a bridge in
        $G_2=H_2+e_2'$, contradicting that $G_2$ is bridgeless. Therefore
        $v_1$ and $v_2$ lie in different components. Assume $v_1\in S_1$ and $v_2\in S_2$.

        It follows that in $G_2$ the only edges between $S_1\cap V(G_2)$ and $S_2\cap V(G_2)$
        are $f_3$ and $e_2'$, so $\{f_3,e_2'\}$ is a $2$-edge-cut of $G_2$.
        By Lemma~\ref{lem:two_edge_cut_matching}, every perfect matching of a bridgeless cubic graph meets a $2$-edge-cut
        in either $0$ or $2$ edges, and hence for $M_2$ we have
        $f_3\in M_2 \quad\Longleftrightarrow\quad e_2'\in M_2.$
        On the $H_1$ side, since $u_1\in S_1$ and $u_2\in S_2$, the added edge $e_1'=u_1u_2$ goes between $S_1$, and $S_2$ in $G_1$. Moreover, the only original edges of $H_1$
        between $S_1$ and $S_2$ are $f_1$ and $f_2$. Thus $\{f_1,f_2,e_1'\}$ is a $3$-edge-cut in
        $G_1$. Since $M_1$ is well-spread in $G_1$, it contains exactly one edge of this cut:
        $|M_1\cap\{f_1,f_2,e_1'\}|=1.$
        
        Finally we compute $|M\cap\{f_1,f_2,f_3\}|$.

        \smallskip
        \noindent\emph{If $e_1'\notin M_1$ and $e_2'\notin M_2$,} then by construction
        $M=M_1\cup M_2$. From $|M_1\cap\{f_1,f_2,e_1'\}|=1$ and $e_1'\notin M_1$ we get
        $|M\cap\{f_1,f_2\}|=1$, and from $e_2'\notin M_2$ and $f_3\in M_2\Leftrightarrow e_2'\in M_2$
        we get $f_3\notin M$. Therefore, $|M\cap\{f_1,f_2,f_3\}|=1$.

        \smallskip
        \noindent\emph{If $e_1'\in M_1$ and $e_2'\in M_2$,} then the construction replaces $e_1',e_2'$
        with $e_1,e_2$, i.e.
        \[
        M=(M_1\cup M_2)\cup\{e_1,e_2\}\setminus\{e_1',e_2'\}.
        \]
        Since $e_1'\in M_1$ and $|M_1\cap\{f_1,f_2,e_1'\}|=1$, we have
        $M\cap\{f_1,f_2\}=\emptyset$. Also $e_2'\in M_2$ implies $f_3\in M_2$, and $f_3$ is not
        removed in the construction, so $f_3\in M$. As $e_1,e_2\notin\{f_1,f_2,f_3\}$ in this case,
        it follows again that $|M\cap\{f_1,f_2,f_3\}|=1$.

        In both subcases, $M$ meets the $3$-edge-cut $\{f_1,f_2,f_3\}$ in exactly one edge.
\end{enumerate}

\end{proof}

\end{document}